\documentclass[draftclsnofoot,onecolumn]{IEEEtran}
\usepackage[english]{babel}
\usepackage{latexsym,amssymb,amsmath,amsthm,amsfonts}
\usepackage{multirow}
\usepackage{bm}

\usepackage[font=small ,labelsep=none]{caption}

\ifCLASSINFOpdf
\else
\fi

\DeclareMathOperator*{\argmin}{argmin}

\newtheorem{proposition}{Proposition}
\newtheorem{corollary}{Corollary}
\newtheorem{claim}{Claim}
\newtheorem{lemma}{Lemma}
\newtheorem{theorem}{Theorem}
\newtheorem{remark}{Remark}

\newcommand{\remove}[1]{}

\begin{document}
%
\title{\huge Bounds and Constructions of Locally Repairable Codes: Parity-check Matrix Approach}

\author{Jie~Hao,
        Shu-Tao~Xia,
        Kenneth W. Shum,
        Bin~Chen, Fang-Wei Fu and Yi-Xian Yang 
\thanks{
This research is supported in part by the Beijing Natural Science Foundation under grant No. 4184093, the National Natural Science Foundation of China under grant Nos. 61801049, 61771273, 61971243 and 61571243, the National Key Research and Development Program of China under Grant 2018YFB1800204, the Nankai Zhide Foundation and the Fundamental Research Funds for Central Universities.}
\thanks{Jie Hao and Yi-Xian Yang are with the Information Security Center, Beijing University of Posts and Telecommunications, Beijing, 100876, China (emails: haojie@bupt.edu.cn, yxyang@bupt.edu.cn).}
\thanks{Shu-Tao~Xia and Bin Chen are with the Shenzhen International Graduate School, Tsinghua University, Shenzhen, 518055, China (emails:  xiast@sz.tsinghua.edu.cn, cb17@mails.tsinghua.edu.cn).}
\thanks{Kenneth W. Shum is with School of Science and Engineering, The Chinese University of Hong Kong
(Shenzhen), Shenzhen, 518172, China (email: wkshum@cuhk.edu.cn).}
\thanks{Fang-Wei Fu is with the Chern Institute of Mathematics and LPMC, and Tianjin Key Laboratory of Network and Data Security Technology, Nankai University, Tianjin, 300071, China (email: fwfu@nankai.edu.cn).}
\thanks{
This paper was presented in part at the 2018 IEEE International Symposium on Information Theory \cite{Hao2018ISIT_LRCMaximalLength},  and the 2016 IEEE International Symposium on Information Theory \cite{Hao2016ISIT_BinaryLRC}, where the upper bound  on the minimum distance of $q$-ary optimal LRCs in Sect. \ref{sec:distance-bound} was presented at \cite{Hao2018ISIT_LRCMaximalLength} and enumerations of optimal binary LRCs in Sect. \ref{sec:binary-lrc} was presented   at \cite{Hao2016ISIT_BinaryLRC}.}
}

\markboth{}%
{Shell \MakeLowercase{\textit{et al.}}: Bare Demo of IEEEtran.cls for IEEE Journals}

\maketitle
\begin{abstract}
A $q$-ary $(n,k,r)$ locally repairable code (LRC) is an $[n,k,d]$ linear code over $\mathbb{F}_q$ such that every code symbol can be recovered by accessing at most $r$ other code symbols.
The well-known Singleton-like bound says that $d \le n-k-\lceil k/r\rceil +2$ and an LRC is said to be optimal if it attains this bound.
In this paper, we study the bounds and constructions of LRCs from the view of parity-check matrices.
Firstly,  a simple and unified framework based on  parity-check matrix to analyze the bounds of LRCs is proposed.
Several useful structural properties on $q$-ary optimal LRCs are obtained.
We derive an upper bound on the minimum distance of $q$-ary optimal $(n,k,r)$-LRCs in terms of the field size $q$.
Then, we focus on constructions of optimal LRCs over binary  field.
It is proved that there are only 5 classes of possible parameters with which optimal binary $(n,k,r)$-LRCs exist.
Moreover, by employing the proposed  parity-check matrix approach, we completely enumerate all these 5 classes of  possible optimal binary  LRCs attaining the Singleton-like bound in the sense of equivalence of linear codes.
\end{abstract}


\IEEEpeerreviewmaketitle

\section{Introduction}

Nowadays, in order to ensure data reliability against storage node failures, redundant data are always stored in large distributed storage systems. Due to the increasing volume of data, the traditional redundancy scheme of 3-replication tends to cause massive storage overhead. Erasure codes are then introduced to reduce the storage overhead while maintaining high data reliability. The most widely used erasure codes are Reed-Solomon codes, which are a class of maximum distance separable (MDS) codes. In such case, the data is firstly divided into $k$ information packets. Then $n-k$ parity packets are generated by encoding these $k$ information packets. Finally, all these $n$ packets are stored in different storage nodes, which can tolerate any $n-k$ failures. 
In case of storage node failures, the storage system needs to repair the failed nodes to maintain data reliability. For 3-replication, when a node fails, node repairing can be accomplished  by reading the data directly from  the replication node. For the case of MDS codes, node repairing involves reading $k$ packets from other nodes, reconstructing the original data from these $k$ packets, and generating the lost packet by encoding the reconstructed data. One can see that the repair cost is much higher than that of 3-replication.

To reduce the repair cost of erasure codes, \emph{locally repairable codes} (LRCs) have emerged in recent years \cite{Gopalan2012IT_LRC,Papailiopoulos2014IT_LRC}.
Consider a $q$-ary $[n,k,d]$ linear code with length $n$, dimension $k$ and minimum distance $d$. A code symbol is said to have \emph{locality} $r$ if it can be repaired by accessing at most $r$ other code symbols.
LRCs are linear codes with locality constraints on code symbols.
An $[n,k,d]$ linear code is called an $(n,k,r)$-LRC if all the code symbols have locality $r$.
For LRCs with locality $r \ll k$, only a small number of storage nodes are  involved in repairing a failed node, which achieves low repair cost  compared to MDS codes.
Windows azure storage employed a class of LRCs as its redundancy scheme \cite{Huang2012ATC_Azure}. The Hadoop Distributed File System RAID used by Facebook implemented another type of  LRCs \cite{Sathiamoorthy2013VLDB_XOR-Elephant}.
Gopalan \emph{et al.}  \cite{Gopalan2012IT_LRC} firstly conducted theoretical analyses of the bounds for LRCs with information locality where only the information symbols satisfy the locality properties. Apparently, the bounds obtained by Gopalan \emph{et al.} also applied to LRCs with all symbol locality. 

For a $q$-ary $(n,k,r)$-LRC,  the following well-known Singleton-like bound was given by Gopalan \emph{et al.}  \cite{Gopalan2012IT_LRC}
 \begin{equation} \label{singleton-like-bound}
 d \leq n -k - \left\lceil \frac{k}{r} \right\rceil + 2.
 \end{equation}
When $r=k$, it reduces to the classical Singleton bound $d \leq n - k + 1$.
Optimal LRCs with small field size are of particular interest.  Many works have proposed  constructions of optimal LRCs attaining the Singleton-like bound (\ref{singleton-like-bound}) over a relatively small field size.
Tamo and Barg \cite{TamoBarg2014IT_RSLRC} proposed an elegant construction of optimal LRCs where  the codewords are constructed by evaluations of specifically designed polynomials over a finite field and the required field size is just $q > n+1$. Specially, when $r=k$, the Tamo-Barg code can reduce to the classical Reed-Solomon code.
Cyclic  optimal  LRCs with $q \ge n+1$ which is characterized  in terms of their zeros were proposed in   \cite{Tamo2015ISIT_CyclicLRC}.
Optimal cyclic LRCs with distance 3 and 4 were proposed by Luo {et al.} \cite{Luo2018IT_CyclicLRC}.
Constructions of optimal LRCs of distance 5 and 6 by using binary constant weight codes were  proposed by Jin \cite{Jin2019IT_LRCd5and6}.
A  refined bound of $(n,k,r)$-LRCs  based on integer programming was proposed by Wang and Zhang \cite{Wang2015IT_IPBound}. 
Taking the field size into account, Cadambe and Mazumdar proposed the following  alphabet-dependent bound of $q$-ary $(n,k,r)$-LRCs  \cite{Cadambe2015IT_DimensionBound},
\begin{equation}\label{general-bound-k}
k \le \min_{ t \in Z_+} \Big[ tr + k^{(q)}_{\rm opt}(n  - t(r+1),d) \Big],
\end{equation}
where $t \le \min \left\{\left\lceil \frac{n}{r+1}\right\rceil,\left\lceil\frac{k}{r}\right\rceil\right\}$, and $k_{\text{opt}}^{(q)}(n',d')$ denotes the largest possible dimension of a $q$-ary linear code with length $n'$ and minimum distance $d'$. Note that the field size $q$ is taken into account,  the Cadambe-Mazumdar bound (\ref{general-bound-k}) is shown to be tighter than the Singleton-like bound  (\ref{singleton-like-bound}), especially when $q$ is small. Optimal LRCs meeting the Cadambe-Mazumdar bound (\ref{general-bound-k}) were proposed in \cite{Silberstein2015ISIT_AnticodeLRC}\cite{Zeh2015ITW_CyclicLRC}\cite{Silberstein2018DCCAnticodeLRC}.
In order to accomplish the local recovery in case of  more than one node failures,  two parallel generalizations of the concepts of locality were proposed.
The first type was  $(r,\delta)$-locality \cite{Prakash2012ISIT_r-deltaLRC} where the code symbol is protected by a punctured subcode with length at most $r+\delta-1$ and minimum distance at least $\delta$. Related bounds and constructions for LRCs with $(r,\delta)$-locality were given in  \cite{Prakash2012ISIT_r-deltaLRC}, \cite{Ernvall2016IT_ConstructionLRC}, \cite{Kamath2014IT_LocalRegeneration}, \cite{Rawat2014IT_Secure}, \cite{Chen2018IT_CyclicDeltaLRC}.
The other type of  generalization is   $(r,t)$-locality \cite{Wang2014IT_r-t-locality}, \cite{Rawat2016IT_r-t-locality}, where the code symbol can be repaired by $t$ disjoint repair groups of the other code symbols, each of size at most $r$.
Bounds and constructions for LRCs with $(r,t)$-locality were given in  \cite{Wang2014IT_r-t-locality}, \cite{Rawat2016IT_r-t-locality}, \cite{Tamo2016IT_r-t-locality}, \cite{Huang2016IT_BinaryLRCAvailability}, \cite{Kruglik2019IT_r-t-LRC}.

In this paper, we study the bounds and constructions of $(n,k,r)$-LRCs from the view of parity-check matrices. Firstly, a simple and unified framework based on  parity-check matrix to study  LRCs is proposed. We set up a  characterization of the parity-check matrix   of a $q$-ary $(n,k,r)$-LRC $\mathcal{C}$  by selecting $n-k$ linearly independent vectors from the dual code $\mathcal{C}^\bot$,  where the first $l$ vectors are locality-rows ensuring the locality properties.
By analyzing this proposed parity-check matrix, we give simple and unified proofs for the Singleton-like bound \eqref{singleton-like-bound}  and  Cadambe-Mazumdar bound \eqref{general-bound-k}.
Based on the new proof technique, the following structural  properties and  bounds of  $q$-ary  optimal $(n,k,r)$-LRCs  attaining the Singleton-like bound \eqref{singleton-like-bound} are obtained.
\begin{itemize}
  \item (see Theorem \ref{thm:subcode-mds}) For a $q$-ary  optimal $(n,k,r)$-LRC, by the new  proof technique, we obtain several useful properties on the structure of parity-check matrices of optimal LRCs. Particularly, we show that after puncturing arbitrary $\lceil k/r \rceil-1$ or $\lceil k/r \rceil-2$ locality-rows from the characterized parity-check matrix $H$ of a $q$-ary optimal LRC, the obtained $q$-ary subcode must be a $q$-ary MDS code or almost MDS code. 
  \item (see Theorem \ref{thm:distance-bound})  For a $q$-ary  optimal $(n,k,r)$-LRC, we derive an upper bound on its  minimum distance  in terms of the field size $q$. It is shown that the minimum distance of a $q$-ary  optimal $(n,k,r)$-LRC must satisfy $d \le 2q$. Specially, when $k-1$ is not divisible by $r$, it holds that $d \le q$. 
\end{itemize}

Then, we focus on constructions of optimal $(n,k,r)$-LRCs   over   binary field. Surprisingly, by employing the proposed parity-check matrix approach, we can completely enumerate all the optimal binary  $(n,k,r)$-LRCs   attaining the Singleton-like bound (\ref{singleton-like-bound}).

\begin{itemize}
\item
  (see Theorem \ref{thm:binary-lrc})  It is  proved that there are only $5$ classes of optimal binary $(n,k,r)$-LRCs with parameters as follows. In the sense of equivalence of linear codes, we completely enumerate all these optimal binary LRCs by presenting their explicit parity-check matrices.
\begin{enumerate}
\item $(k+k/r,k,r)$, $d=2$, $k>r\ge 1$, $r\mid k$;

\item $(k+\lceil k/r\rceil,k, r)$, $d=2$, $k>r\ge 1$, $r\nmid k$;

\item $(2k+2,k,1)$, $d=4$, $k\ge 2$;

\item $(4l,3l-2,3)$, $d=4$, $l\ge 3$;

\item  $(k+d,k,k-1)$, $3 \le d \le 4, 3\le k \le 4$.
\end{enumerate}

\end{itemize}

The rest of this paper is organized as follows. Section~\ref{sec:preliminaries} gives some notations and preliminaries.
In Section~\ref{sec:PCM}, we propose the parity-check matrix approach  to  analyze the bounds of LRCs.
Section~\ref{sec:property} presents several  structural properties of $q$-ary optimal LRCs.
Section~\ref{sec:distance-bound} derives an upper bound on the minimum distance of $q$-ary optimal LRCs.
In Section~\ref{sec:binary-lrc}, we  enumerate all the  possible  optimal binary LRCs. 
Section \ref{sec:conclusion} concludes the paper.

\section{Notations and Preliminaries} \label{sec:preliminaries}

Let ${\mathbb{F}}_q$ be a finite field with size $q$, where $q$ is a prime power.
Let $\mathcal{C}$ be a $q$-ary $[n,k,d]$ linear code with length $n$, dimension $k$ and minimum distance $d$ \cite{MacWilliamsCodingBook}.
The $k\times n$ generator matrix $G$ and  $(n-k)\times n$ parity-check matrix $H$ satisfies $G H^T = 0$.
Let $\mathcal{C}^\bot$ denote the dual code of $\mathcal{C}$.
The rows of $H$ are the codewords of $\mathcal{C}^\bot$, and are also called parity-check equations.
Let $A \otimes B$ denote the Kronecker product of matrices $A$ and $B$.
Let $\mathbf{a}=(a_1,a_2,\ldots, a_n)$ be a vector and $[n] = \{1,2,\cdots,n\}$.
The support of a vector $\mathbf{a}$ is $supp(\mathbf{a}) = \{i \in [n]: a_i \neq 0\}$, and its (Hamming) weight is $wt(\mathbf{a}) = |supp(\mathbf{a})|$. If the index of a coordinate is in the support of a vector, it is said to \emph{be covered} by the vector.
The (Hamming) distance of two vectors is the number of coordinates at which they differ. The minimum distance $d$ of $\mathcal{C}$ is the minimum value of distances between any two different codewords.
For a $q$-ary $[n,k,d]$ linear code  $\mathcal{C}$, the classical  {\em Singleton bound} says that $d \leq n - k + 1$, where the equality holds when it is an \emph{maximum distance separable (MDS)} code.

For a $q$-ary $[n,k,d]$ linear code $\mathcal{C}$, the {\em Singleton defect} of $\mathcal{C}$ is defined as $\lambda(\mathcal{C}) := n-k-d+1.$
A linear code $\mathcal{C}$ with $\lambda(\mathcal{C})=0$ is precisely an MDS code. A linear code $\mathcal{C}$
with $\lambda(\mathcal{C})=1$ is called an {\em almost MDS code} \cite{deBoer1996DCC_AMDS}.
For a nontrivial $q$-ary $[n,k,d]$ code with defect $\lambda(\mathcal{C})= \lambda,$ the following lemma holds.

\begin{lemma}[\cite{Faldum97DCC_Defect}] \label{lemma:defect}
A $q$-ary $[n,k,d]$ linear code $\mathcal{C}$ with dimension $k\ge 2$ and Singleton defect $\lambda(\mathcal{C}) = \lambda$ has minimum distance $d \le q(\lambda+1)$.
\end{lemma}

An almost MDS code $\mathcal{C}$ satisfying the further condition that $\lambda(\mathcal{C}^\perp)=1$ is called a {\em near MDS code}~\cite{Dodunekov1995_nearMDS}.
The dual code of an $[n,k,n-k]$ near MDS code is an  $[n,n-k,k]$  near MDS code \cite{Dodunekov1995_nearMDS}. When $q=2$, unlike binary MDS codes which are all trivial linear codes with parameters $[n,n-1,2]$ or $[n,1,n]$, there exist some nontrivial binary near MDS codes. Since both the $[n,k,n-k]$ near MDS code and its $[n,n-k,k]$ dual code  have defect $\lambda=1$, by Lemma \ref{lemma:defect}, we have $n-k\le 2q = 4$, $k\le 2q=4$, and $n\le 8$ for binary near MDS codes.
In fact, the next result holds for binary near MDS codes.
\begin{lemma} [\cite{Dodunekov1995_nearMDS}]\label{binary-near-mds}
When $k\ge 3$ and  $d \ge 3$, up to the equivalence of linear codes, there only exist four nontrivial binary $[n,k]$ near MDS codes, i.e., the binary $[7,4,3]$  Hamming code, the binary $[8,4,4]$  extended Hamming code, the binary $[7,3,4]$  Simplex code and the binary $[6,3,3]$ punctured Simplex code.
\end{lemma}

\section{Parity-check Matrix Approach for LRCs} \label{sec:PCM}

In this section, we propose a parity-check matrix approach to analyze the bounds of $q$-ary $(n,k,r)$-LRCs with all symbol locality. Firstly, we set up a  characterization of the parity-check matrix for an $(n,k,r)$-LRC. Then, by analyzing the characterized parity-check matrix, we give a simple and unified proof of the Singleton-like bound (\ref{singleton-like-bound}) and  Cadambe-Mazumdar bound (\ref{general-bound-k}). The new  proof approach reveals the connections of these two bounds.  It is shown  that the  Singleton-like bound (\ref{singleton-like-bound}) is a special case of the Cadambe-Mazumdar bound (\ref{general-bound-k}).

\subsection{Characterization of  Parity-check Matrix }\label{cpm}

Suppose that $\mathcal{C}$ has locality $r$, or consider $\mathcal{C}$ as a $q$-ary $(n,k,r)$-LRC with all symbol locality. By choosing $n-k$ linearly independent codewords from the dual code $\mathcal{C}^\bot$, we can obtain a full-rank parity-check matrix of $\mathcal{C}$.
The locality property of an LRC can  be characterized by the parity-check matrix. In order to find a suitable parity-check matrix to involve locality, we begin with a simple observation:

\begin{claim}\label{claim:locality}
A code symbol has locality $r$ if and only if there exists a codeword in  $\mathcal{C}^\bot$ which has at most $r+1$ non-zero components and covers the coordinate of this symbol.
\end{claim}

According to Claim \ref{claim:locality}, we can select $n-k$ independent codewords from $\mathcal{C}^\bot$ to form a  parity-check matrix $H$ of $\mathcal{C}$ involving locality properties, which is divided into two parts 
\begin{equation}\label{PCM}
H = \left[
      \begin{array}{c}
        H_1 \\
        H_2 \\
      \end{array}
    \right].
\end{equation}
The rows in the upper part $H_1$, or \emph{locality-rows}, cover all the $n$ coordinates and ensure locality. 
All locality-rows in $H_1$ have weight at most $r+1$.
As for the selection procedure, firstly, for the first coordinate, select a codeword from $\mathcal{C}^\bot$ with weight at most $r+1$ to cover it; then, for the first uncovered coordinate, select another codeword from $\mathcal{C}^\bot$  with weight at most $r+1$ to cover it; repeating the procedure iteratively until all the $n$ coordinates are covered and $H_1$ is constructed. Let $l$ be the number of rows in $H_1$ (or the number of locality-rows). Clearly, these $l$ rows are linearly independent. Then, we select some other  $n-k-l$ independent codewords  from $\mathcal{C}^\bot$ to form the lower part $H_2$, and the construction of $H$ completes.
The details are given as follows.

\begin{center}
\fbox{%
\parbox{13.3cm}{%

1.\ \ Let $i =1, S_0 = \{\}$.\quad\quad\quad\quad\quad\quad\quad\quad\quad // \small{initialization.}

2.\ \ While $S_{i-1} \neq [n]$:

3.\ \ \hspace{0.5cm} Pick $j \in [n] \setminus S_{i-1}$. \quad\quad\quad\quad\quad\quad\quad\quad// \small{pick a coordinate $j$ not covered.}

4.\ \ \hspace{0.5cm} Choose $\mathbf{h}_i = \mathop{\argmin}_{\mathbf{e} \in \mathcal{C}^\bot, \;e_j \neq 0} \;{wt(\mathbf{e})}.$ \;\;// \small{find a  codeword from $\mathcal{C}^\bot$  covering $j$.}

5.\ \ \hspace{0.5cm} Set $S_i = S_{i-1} \bigcup supp(\mathbf{h}_i)$.\quad\quad\quad\quad\quad\;\;// \small{the set of coordinates covered by the first $i$ rows. }

6.\ \ \hspace{0.5cm} $i=i+1.$

7.\ \ Set $l = i-1$. Set $H_1 = \left[\begin{array}{c} \mathbf{h}_1 \\ \vdots \\ \mathbf{h}_l \end{array} \right]$.

8.\ \ Choose additional $n-k-l$ vectors from $\mathcal{C}^\bot$  such that $H_2 = \left[\begin{array}{c} \mathbf{h}_{l+1} \\ \vdots \\ \mathbf{h}_{n-k} \end{array} \right]$ and  $H = \left[\begin{array}{c} H_1 \\  H_2\end{array} \right]$\\

\ \ \ \     is an $(n-k)\times n$ full-rank matrix.

}%
}
\end{center}
In the line 4 of the $i$-th iteration, by the above Claim 1, such a codeword exists and covers at most $r+1$ symbols. Moreover, the $i$-th row covers some coordinates not covered by previous ones, which implies that it is linearly independent with them. Repeat the choosing procedure to get $l$ independent codewords in $H_1$. Clearly, $l\le n-k$ or $l+k\le n$. Moreover, since each of the $l$ rows has weight at most $r+1$, $n\le l(r+1)$, which implies $l+k\le l(r+1)$ or $k/r\le l$. Thus, $l+k\le n$ implies $k/r+k\le n$ or $k/r\le n/(r+1)$. Combining these, we have
\begin{eqnarray}
\label{eql}
\frac{k}{r}\;\le\; \frac{n}{r+1} &\le&l\;\le\; n-k. 
\end{eqnarray}
In the rest of the paper, the rows in $H_1$ are called \emph{locality-rows}.
Since the number of the locality-rows is $l\le n-k$, line 8 is always feasible.

\subsection{Unified Proof for Several Different Bounds of LRCs}

By analyzing the characterized  parity-check matrix, we   present a unified proof for the Singleton-like bound  and several  bounds concerning the field size, including the well-known Cadambe-Mazumdar bound.
This offers a new way to understand different bounds of LRCs from a viewpoint of parity-check matrix.

\smallskip

\begin{proposition}\label{prop:general-bound-distance}
For a $q$-ary $(n,k,r)$-LRC $\mathcal{C}$ with all symbol locality, the minimum distance satisfies
\begin{equation} \label{general-bound-d}
d \le \min_{ 1 \leq \tau \leq \left\lceil\frac{k}{r}\right\rceil-1} \; d^{(q)}_{\rm opt}(n  - \tau(r+1),k-\tau r),
\end{equation}
where $d_{\text{opt}}^{(q)}(n^*,k^*)$ is the largest  minimum distance of a $q$-ary linear code with length $n^*$ and dimension  $k^*$.
\end{proposition}

\begin{proof}
Let $H$ be the proposed parity-check matrix of $\mathcal{C}$ in Section \ref{cpm}. By (\ref{eql}), we know $l \ge \left\lceil\frac{k}{r}\right\rceil$. Consider the first $\tau$ locality-rows of $H_1$, where $1 \leq \tau \leq \left\lceil\frac{k}{r}\right\rceil-1$. Let $\gamma$ be the number of the columns that the non-zero entries of these $\tau$ locality-rows lie in. Then the locality property implies $\gamma \leq \tau(r+1)$. By removing the first $\tau$ locality-rows and the corresponding $\gamma$ columns of $H$, and further removing $\tau(r+1)-\gamma$ columns, we have an $m^*\times n^*$ submatrix $H^*$, where $m^*= n-k-\tau$ and $n^* = n - \tau(r+1)$. Let $\mathcal{C}^*$ be the $[n^*,k^*,d^*]$ linear code with $H^*$ as parity-check matrix.
Among the corresponding $n^*$ columns of $H$, since the entries above $H^*$ are all zeros, $d \leq d^*$.
Moreover, by ${\rm rank}(H^*)\leq  n-k-\tau$, we have  $k^* = n^*-{\rm rank}(H^*) \geq   k-\tau r>0$. Hence,
\begin{eqnarray}
 d\;\le \;d^*  \;\leq\; d^{(q)}_{\rm opt}(n^*,k^*) \;\leq\; d^{(q)}_{\rm opt}(n - \tau(r+1),k-\tau r). \label{C*}
\end{eqnarray}
Since  $1 \leq \tau \leq \left\lceil\frac{k}{r}\right\rceil-1$, the conclusion follows.
\end{proof}

Let  $n_{\text{opt}}^{(q)}(k',d')$ be the smallest length of a $q$-ary linear code with dimension $k'$ and minimum distance $d'$. Let  $k_{\text{opt}}^{(q)}(n',d')$ be the largest  dimension of a $q$-ary linear code with length $n'$ and minimum distance  $d'$.
Using the similar argument in the proof of Proposition \ref{prop:general-bound-distance},  by substituting (\ref{C*}) with
\begin{equation*}
 n =  \tau(r+1) + n^* \geq \tau(r+1) + n_{\text{opt}}^{(q)}(k^*,d^*) \geq  \tau(r+1) +  n^{(q)}_{\rm opt}(k-\tau r, d),
\end{equation*}
or
\begin{eqnarray*}
k\;\le \;k^*+ \tau r  \;\leq\; k^{(q)}_{\rm opt}(n^*,d^*)+\tau r \;\leq\; k^{(q)}_{\rm opt}(n - \tau(r+1),d)+\tau r,
\end{eqnarray*}
we can obtain that
\begin{equation} \label{general-bound-length}
n \geq \max_{ 1 \leq \tau \leq \left\lceil\frac{k}{r}\right\rceil-1} \; \left[ \tau(r+1) +  n^{(q)}_{\rm opt}(k-\tau r, d) \right],
\end{equation}
or
\begin{equation}\label{CM-bound}
k \le \min_{  1 \leq \tau \leq \left\lceil\frac{k}{r}\right\rceil-1} \left[ \tau r + k^{(q)}_{\rm opt}(n  - \tau(r+1),d) \right],
\end{equation}
which is exactly the Cadambe-Mazumdar bound.

Consider the proof in  Proposition \ref{prop:general-bound-distance},
when $\tau=\left\lceil\frac{k}{r}\right\rceil-1$ and invoking the Singleton bound $$ d \le d_{\text{opt}}^{(q)} (n  - \tau(r+1),k-\tau r)\le n-k-\tau +1=n-k-\left\lceil\frac{k}{r} \right\rceil + 2,$$ the Singleton-like bound (\ref{singleton-like-bound})  follows naturally. Thus, this gives unified proof of the bound \eqref{general-bound-d} and the Singleton-like bound (\ref{singleton-like-bound}).
Similarly, for  the bounds \eqref{general-bound-length} and \eqref{CM-bound}, when $\tau=\left\lceil\frac{k}{r}\right\rceil-1$ and invoking the Singleton bound
$$n^{(q)}_{\rm opt}(k-\tau r, d) \ge d + k-\tau r -1 \;\; \mbox{ or } \;\; k^{(q)}_{\rm opt}(n  - \tau(r+1),d) \le n  - \tau(r+1) -d+1,$$
the Singleton-like bound (\ref{singleton-like-bound}) can also be obtained.
Clearly, the general bounds (\ref{general-bound-d}) and (\ref{general-bound-length})  are essentially identical to the Cadambe-Mazumdar bound (\ref{CM-bound}).
Since the field size is taken into account, these  general bounds can yield better results than the Singleton-like bound over small fields.

\smallskip
Different bounds can be obtained from the general bound (\ref{general-bound-d}) and (\ref{general-bound-length}) by choosing different bounds of $d^{(q)}_{\rm opt}(n^*,k^*)$ or $ n_{\text{opt}}^{(q)}(k^*,d^*)$. For example, by  applying  the Plotkin bound \cite{MacWilliamsCodingBook}, and the Griesmer bound \cite{HuffmanCodingBook}, the following two bounds can be obtained.

\medskip
\begin{corollary}[Plotkin-like bound]\label{Coro:Plotkin-like-bound}
For a $q$-ary $(n,k,r)$-LRC with all symbol locality,
\begin{equation} \label{Plotkin-like-bound}
d \le \min_{ 1 \leq \tau \leq \left\lceil\frac{k}{r}\right\rceil-1} \; \frac{q^{k-\tau r-1}(q-1)[n  - \tau(r+1)]}{q^{k-\tau r}-1}.
\end{equation}
\end{corollary}

\medskip
\begin{corollary}[Griesmer-like bound]\label{Coro:Griesmer-bound}
For a $q$-ary $(n,k,r)$-LRC with all symbol locality,
\begin{equation} \label{Griesmer-like bound}
n \geq \max_{ 1 \leq \tau \leq \left\lceil\frac{k}{r}\right\rceil-1} \; \left\{ \tau(r+1) + \sum_{i=0}^{k-\tau r-1} \left\lceil\frac{d}{q^{i}}\right\rceil \right\}.
\end{equation}
\end{corollary}

\smallskip
Taking $\tau= \left\lceil\frac{k}{r}\right\rceil-1$ in the Griesmer-like bound (\ref{Griesmer-like bound}), we obtain the following lower bound on the code length of  a $q$-ary $(n,k,r)$-LRC.

\begin{corollary}  \label{Coro:TKR-Griesmer-bound}
For a $q$-ary $(n,k,r)$-LRC with all symbol locality, the code length satisfies
\begin{equation}  \label{tkr-griesmer-like-bound}
 n \geq (r+1)\left(\left\lceil\frac{k}{r}\right\rceil-1 \right) +  \sum_{i=0}^{k+r-1-r\left\lceil\frac{k}{r}\right\rceil } \left\lceil\frac{d}{q^{i}}\right\rceil.
\end{equation}

\end{corollary}
The Singleton-like bound (\ref{singleton-like-bound}) also follows from the bound (\ref{tkr-griesmer-like-bound}) since
\begin{eqnarray*}
& & (r+1)\left(\left\lceil\frac{k}{r}\right\rceil-1 \right) +  \sum_{i=0}^{k+r-1-r\left\lceil\frac{k}{r}\right\rceil } \left\lceil\frac{d}{q^{i}}\right\rceil  \\
&=& (r+1)\left(\left\lceil\frac{k}{r}\right\rceil-1 \right) + d + \sum_{i=1}^{k+r-1-r\left\lceil\frac{k}{r}\right\rceil } \left\lceil\frac{d}{q^{i}}\right\rceil  \\
&\geq &  (r+1)\left(\left\lceil\frac{k}{r}\right\rceil-1 \right) + d + k+r-1-r\left\lceil\frac{k}{r}\right\rceil  \\
&=&  d + k + \left\lceil\frac{k}{r} \right\rceil -2.
\end{eqnarray*}
The binary $[15,11,3]$ Hamming code which has locality $r=7$ and the binary $[23,12,7]$ Golay code with $r=7$ attain the bound (\ref{tkr-griesmer-like-bound}) with equality, which certifies the tightness of bound (\ref{tkr-griesmer-like-bound}). Note that these two binary linear codes do not attain the Singleton-like bound (\ref{singleton-like-bound}).

\medskip
\section{Properties of  $q$-ary Optimal LRCs Attaining the Singleton-like bound} \label{sec:property}

In this section,  firstly, we proposed  an alternative simple and refined proof for the Singleton-like bound \eqref{singleton-like-bound} by analyzing the characterized parity-check matrix $H$ in Section \ref{cpm}.
It is shown that the parity-check matrix $H$ must have  $n-k-\left\lceil\frac{k}{r} \right\rceil + 2$ linearly dependent columns, since their nonzero entries lie in at most $n-k-\left\lceil\frac{k}{r} \right\rceil + 1$ rows.
Then, by using the new  proof technique, we obtain several useful  structural properties of  $q$-ary optimal $(n,k,r)$-LRCs attaining the Singleton-like bound.

\smallskip
\begin{proposition} [Singleton-like bound \cite{Gopalan2012IT_LRC}]\label{prop:singleton-bound}
For an $(n,k,r)$-LRC with all symbol locality, the minimum distance
    $d\leq n-k- \left\lceil\frac{k}{r} \right\rceil + 2.$
\end{proposition}

\begin{proof}
It is enough to show the proposed parity-check matrix $H$ in Section \ref{cpm} must have $n-k-\left\lceil\frac{k}{r} \right\rceil + 2$ linearly dependent columns.
By (\ref{eql}), the number of the locality-rows is $l \ge \left\lceil\frac{k}{r}\right\rceil$.
Now consider the first $\tau=\left\lfloor\frac{k}{r}\right\rfloor$ locality-rows of $H_1$. Let $\gamma$ be the number of the columns that the non-zero entries of these $\tau$ locality-rows lie in. Then the locality property implies $\gamma \leq \tau(r+1)$. The number of the remaining columns is $n - \gamma \geq n - \tau(r+1)$, where the equality holds if and only if the supports of the first $\tau$ locality-rows are pairwise disjoint and each has weight exactly $r+1$. The number of the remaining rows is $\eta=n-k-\tau$.

\textbf{Case 1}: If $r\nmid k$,  then $n-\gamma\ge n - \tau (r+1) > n-k-\tau=\eta$, i.e.,
\begin{equation}
n-\gamma\ge \eta+1 = n - k - \left\lfloor\frac{k}{r}\right\rfloor +1 = n - k - \left\lceil\frac{k}{r}\right\rceil +2. \label{conclusion1}
\end{equation}
The first $\eta+1$ columns in the remaining $n-\gamma$ columns must be linearly dependent since the non-zero entries of these columns are contained in only $\eta$ rows. This implies that $d\le \eta+1 = n - k - \lceil k/r \rceil +2.$

\textbf{Case 2}: If $r\mid k$, then $n-\gamma\ge n - \tau(r+1) = n-k-\tau=\eta.$
If $n-\gamma\ge \eta+1$, we have $d\le \eta+1$ with similar arguments to Case 1. Otherwise, if $n-\gamma = \eta$, then the supports of the first $\tau$ locality-rows are pairwise disjoint and each has weight exactly $r+1$. Choose two columns from the support of the first locality-row, and combine these two columns with the remaining $\eta$ columns, we have $\eta+2$ columns.
These $\eta+2$ columns have their non-zero entries existing in only $\eta+1$ rows, and thus are linearly dependent. This implies that $d\le \eta+2=n-k-\frac{k}{r}+2$.

Combining the above two cases, the conclusion follows.
\end{proof}

By using the above  proof technique based on parity-check matrix, we can obtain the following three   properties on the structures of  $q$-ary  optimal $(n,k,r)$-LRCs attaining the Singleton-like bound \eqref{singleton-like-bound}.

\smallskip
\begin{theorem} \label{thm:subcode-mds}
Suppose that $k>r\geq 1$. Let $\mathcal{C}$ be a $q$-ary optimal   $(n,k,r)$-LRC with $d=n-k-\lceil{k}/{r}\rceil+2$ and $H$ be its parity-check matrix constructed in Section \ref{cpm}.
Let $H'$ be an $m'\times n'$ submatrix obtained from $H$ by removing  any fixed  $\left\lceil{k}/{r}\right\rceil-1$ locality-rows and all the columns whose coordinates are covered by the supports  of these $\left\lceil{k}/{r}\right\rceil-1$ locality-rows.
Let $H''$ be an $m''\times n''$  submatrix obtained from $H$ by removing any fixed $\lceil k/r \rceil- 2$ locality-rows and all the associated columns.   Then,
\begin{enumerate}
  \item  $H'$ has full rank and $m'=d-1$. The $[n',k',d']$ linear code $\mathcal{C}'$ with $H'$ as  parity-check matrix is a $q$-ary MDS code with $d'=d$.
  \item   $H''$ also has full rank and $m'' = d$. The $[n'',k'',d'']$ linear code $\mathcal{C}''$ with $H''$ as  parity-check matrix is a $q$-ary almost MDS code with $d'' = d$.
\end{enumerate}

\end{theorem}

\begin{proof}
For the first part of the theorem, by (\ref{eql}),  $H$ contains $l \ge \left \lceil k/r \right \rceil$ locality-rows.   After removing  any fixed $\left \lceil k/r  \right \rceil-1$ locality-rows from $H$, the number of the remaining rows is
\begin{eqnarray}
m'=n - k - \left \lceil \frac{k}{r} \right \rceil + 1\ge 1,\label{eqm'}
\end{eqnarray}
and $H'$ has at least one row with weight at most $r+1$, which has ever been a locality-row of $H$ before removing.
Let $\gamma$ be the number of the columns covered by the removed $\left\lceil{k}/{r}\right\rceil-1$ locality-rows. Since every locality-row has weight at most $r+1$, we have $\gamma \leq (\left\lceil{k}/{r}\right\rceil-1)(r+1)$.
Combining $(\left\lceil{k}/{r}\right\rceil-1) \cdot r<k$, we have
\begin{equation}
  n'\geq n - \left(\left \lceil \frac{k}{r} \right \rceil-1\right )(r+1) > n - k - \left \lceil \frac{k}{r} \right \rceil + 1 = m'\ge 1.
\end{equation}
Consider the $[n',k',d']$ linear code $\mathcal{C}'$ with $H'$ as  parity-check matrix, by the classical Singleton bound,
\begin{equation} \label{H*1}
     d' \le n'- k' + 1 = {\rm Rank}(H') + 1 \le m' + 1.
\end{equation}
Among these $n'$ columns of $H$, since the entries above $H'$ are all zeros, we have $d \leq d'$.
Since $\mathcal{C}$ is an optimal LRC with $d = n - k - \left\lceil{k}/{r}\right\rceil + 2$, we obtain
\begin{equation} \label{H*2}
     d' \geq d = n - k -\left \lceil \frac{k}{r} \right \rceil + 2 = m' +1.
\end{equation}
Combining (\ref{H*1}) and (\ref{H*2}), we have $d'=m'+1$. Hence, all equalities in   (\ref{H*1}) and   (\ref{H*2}) hold. Therefore, $m'={\rm Rank}(H')=n'-k'=d-1$ and   $\mathcal{C}'$ is a $q$-ary MDS code with $d' = d$.

\smallskip

For the second part of the theorem, firstly we observe that $\lceil k/r \rceil- 2$ is nonnegative, as $k>r$ by assumption, and $H'' = H$ when $\lceil k/r \rceil = 2$.
The number of rows in $H''$ is
\begin{equation}
  m'' = n-k-\left \lceil \frac{k}{r} \right \rceil+ 2.
\end{equation}
According to the first part of this theorem, by removing one more locality-row and all the columns covered by this locality-row from $H''$, the resulting submatrix $H'$ has full rank, and the $[n',k',d']$ linear code with $H'$ as parity-check matrix has minimum distance $d'=d$.
Among these $n''$ or $n'$ columns of $H$ which correspond to the columns of $H''$ or $H'$, since the entries above $H''$ or $H'$ are all zeros,
\begin{equation}
  d \le d'' \le d' = d.
\end{equation}
Therefore, $d'' = d$.
Since $H'$ has full rank, it is not hard to see that $H''$ also has full rank. Hence,
\begin{equation}
  n''-k'' = m'' \mbox{ and }  d''=d = m''.
\end{equation}
The Singleton defect is $n''-k''-d''+1 = 1$. Therefore, $\mathcal{C}''$ is a $q$-ary almost MDS code with $d'' = d$.

Combining all the above discussions of two cases, the theorem follows.
\end{proof}

\smallskip
\begin{lemma}\label{lemma:structural-condition}
Suppose that $k>r\geq 1$. Let $\mathcal{C}$ be a $q$-ary optimal   $(n,k,r)$-LRC with $d=n-k-\lceil{k}/{r}\rceil+2$ and $H$ be its parity-check matrix constructed in Section \ref{cpm}.
\begin{itemize}

\item If $r \mid k$, then $(r+1) \mid n$ and the supports of the locality-rows in the parity-check matrix $H$ must be pairwise disjoint, and every locality-row has weight exactly $r+1$.
\item If $r \nmid k$, then the supports of any $ \left\lceil k/r \right\rceil$ locality-rows in the parity-check matrix  $H$ cover at least $ k+ \left\lceil  k/r \right\rceil$ coordinates.
\end{itemize}
\end{lemma}

\smallskip
\begin{proof}
The proofs are divided into two cases as follows.
\begin{itemize}
\item
If $r\mid k$, then $\tau=\frac{k}{r}\ge 2$.
Consider the first $\tau=\frac{k}{r}$ locality-rows of $H_1$. Let $\gamma$ be the number of the columns covered by these $\tau$ locality-rows. Then the locality property implies $\gamma \leq \tau(r+1)$, which indicates that the number of the remaining columns is $n - \gamma \geq n - \tau(r+1)$, where the equality holds if and only if the supports of the first $\tau$ locality-rows are pairwise disjoint and each has weight exactly $r+1$. The number of the remaining rows is $\eta=n-k-\tau$ and  $d=n-k-\frac{k}{r}+ 2=\eta+2$.
For $\tau=\frac{k}{r}\ge 2$, it follows that  $n-\gamma\ge \eta$.
If $n-\gamma\ge \eta+1$, then the first $\eta+1$ columns in the remaining $n-\gamma$ columns must be linearly dependent since the non-zero entries of these columns are contained in only $\eta$ rows. This implies that $d\le \eta+1$.
Thus, we have that $n-\gamma=\eta=n-k-\tau=n-\tau(r+1)$. So the supports of the first $\tau$ locality-rows are pairwise disjoint and each has weight exactly $r+1$. It is easy to see that if we choose any fixed $\tau=\frac{k}{r}$ locality-rows of $H_1$, the same arguments still hold. Hence, we have that the supports of any fixed $\tau$ locality-rows are pairwise disjoint and each has weight exactly $r+1$, which implies the supports of all locality-rows in $H_1$ are pairwise disjoint and each has weight exactly $r+1$, which implies that $(r+1)\mid n$.

\item
If $r\nmid k$. Assume the contrary that there are $\left\lceil k/r\right\rceil$ locality-rows whose nonzero entries cover less than $ k+ \left\lceil k/r \right\rceil$ columns, then the number of remaining columns is greater than $n - k- \left\lceil k/r \right\rceil$, the number of remaining rows is $ n - k -\left\lceil k/r \right\rceil$. There must have $n - k -\left\lceil k/r \right\rceil+1$ columns which are linearly dependent since the non-zero entries of these columns are contained in only $n - k -\left\lceil k/r \right\rceil$ rows, thus $d \leq n - k -\left\lceil k/r \right\rceil+1$, which leads to a contradiction.
\end{itemize}
\end{proof}

\begin{remark}
Lemma \ref{lemma:structural-condition} is identical to the conclusions in \cite[Theorem 2, Corollary 1]{Tamo2013IT_Matroid}.
For the sake of completeness, we include this lemma here by deriving it using this  parity-check matrix approach.
\end{remark}

\medskip

\begin{lemma} \label{lemma:d-dual-r+1}
Suppose that $k>r\geq 1$.  Let $\mathcal{C}$ be a $q$-ary optimal  $(n,k,r)$-LRC with $d=n-k-\lceil{k}/{r}\rceil+2$. If $r \mid k$, then  the dual code $\mathcal{C}^{\bot}$ has  minimum distance   $d(\mathcal{C}^{\bot}) = r+1$.
\end{lemma}
\begin{proof}
 When $r\mid k$, assume the contrary that there is a codeword with weight less than $ r+1$ in the dual code $\mathcal{C}^{\bot}$, then this codeword can be chosen as the locality-row in the parity-check matrix $H$ of $\mathcal{C}$. This contradicts with the first part of Lemma \ref{lemma:structural-condition} that all locality-rows in $H$ must  have unform weight $r+1$. Thus, the minimum weight of the codewords in $\mathcal{C}^{\bot}$ is $r+1$. Therefore, the minimum distance of  $\mathcal{C}^{\bot}$ is  $d(\mathcal{C}^{\bot}) = r+1$.
\end{proof}

\medskip

\section{Upper Bounds on the Minimum Distance of $q$-ary  Optimal LRCs}\label{sec:distance-bound}

Given a $q$-ary optimal  $(n,k,r)$-LRC attaining the Singleton-like bound \eqref{singleton-like-bound} with fixed field size $q$,  an upper bound on its minimum distances  in terms of $q$ is derived in this section. This bound also corresponds to an upper bound on the maximal code length of $q$-ary optimal  $(n,k,r)$-LRCs.

\begin{theorem} \label{thm:distance-bound}
Let $k > r \geq 1$ and $d > 2$. For a $q$-ary optimal   $(n,k,r)$-LRC $\mathcal{C}$  with $d=n-k-\lceil{k}/{r}\rceil+2$, its minimum distance is upper bounded by
\begin{equation}
  d \leq
\begin{cases}
  q,  & \mbox{if } r \nmid (k-1), \\
  2q,  & \text{if } r \mid (k-1). \\
\end{cases}
\end{equation}
\end{theorem}

\begin{proof}
Let $H$ be the parity-check matrix of  $\mathcal{C}$ constructed in Section \ref{cpm}.
Let $H'$ be the submatrix obtained from  $H$ by removing any fixed $\lceil \frac{k}{r} \rceil-1 = \lfloor \frac{k-1}{r} \rfloor$ locality-rows and the columns whose coordinates are covered by these removed locality-rows.
By Theorem \ref{thm:subcode-mds}, $H'$ has full rank and the $q$-ary $[n',k',d']$ linear code $\mathcal{C}'$ with $H'$ as parity-check matrix is a $q$-ary MDS code with $d' = d$.
The number of the remaining rows in $H'$ is $$m' = {\rm Rank}(H') = n - k - \left\lfloor \frac{k-1}{r} \right\rfloor.$$
Let $\gamma$ be the number of the columns covered by these  removed $\lfloor \frac{k-1}{r} \rfloor$ locality-rows.
Then, $\gamma  \le \lfloor \frac{k-1}{r} \rfloor(r+1).$
The number of the columns in $H'$ is $n' =  n - \gamma  \ge n - \lfloor \frac{k-1}{r} \rfloor(r+1).$
Therefore, $\mathcal{C}'$ has dimension
\begin{eqnarray}
 k' = n' - {\rm Rank}(H')   \ge   k- \left \lfloor \frac{k-1}{r} \right \rfloor \cdot r \ge 1. \label{k'-condition}
\end{eqnarray}

Then, we distinguish two cases of $r \nmid (k-1)$ and $r \mid (k-1)$.

\smallskip

\textbf{Case 1}: Suppose that $r \nmid (k-1)$, we have
\begin{equation}
  k' \ge k- \left \lfloor \frac{k-1}{r}  \right \rfloor \cdot r > 1.
\end{equation}
Hence, $C'$ has dimension $k' \ge 2$. The MDS code $C'$ has defect $\lambda(C')=0$. By Lemma~\ref{lemma:defect}, the minimum distance of $\mathcal{C}'$  satisfies $d' \le q$. Then, $\mathcal{C}$ has  minimum distance   $d = d' \le q$.

\textbf{Case 2}: Suppose that $r \mid (k-1)$, we have
\begin{equation}
  k' \ge k- \left \lfloor \frac{k-1}{r}  \right \rfloor \cdot r = 1.
\end{equation}
If $k' \ge 2$, we have $d = d' \le q$ with similar arguments to Case 1.
If $k' = 1$, then $C'$ is an MDS code with dimension $1$.
Now we remove any fixed $\lceil k/r \rceil- 2$ locality-rows of $H$ and the  columns associated with the coordinates covered by these $\lceil k/r \rceil- 2$ locality-rows.
By the second part of Theorem~\ref{thm:subcode-mds}, the resulting submatrix $H''$ has full rank, and the  $[n'',k'',d'']$ linear code $\mathcal{C}''$  with $H''$  as parity-check matrix is a $q$-ary   almost MDS code with $d''=d$. The number of the rows in $H''$ is $m'' = {\rm Rank}(H'') = n - k - (\lceil k/r \rceil-2).$ The code length of $\mathcal{C}''$  satisfies $n'' \ge n - (\lceil k/r \rceil - 2)(r+1)$. Then, the dimension of $\mathcal{C}''$ is
\begin{align*}
 \phantom{=}  k'' = n'' - {\rm Rank}(H'')   \ge \left [n - \left(\left \lceil \frac{k}{r} \right \rceil - 2 \right)(r+1) \right ] - \left[n-k- \left(\left \lceil \frac{k}{r} \right \rceil-2\right)\right] = -\left \lceil \frac{k}{r} \right \rceil r +2r  + k  > 1.
\end{align*}
The almost MDS code $C''$ has defect $\lambda(C'')=1$. By Lemma~\ref{lemma:defect},  the minimum distance of $\mathcal{C}''$ satisfies $d'' \le 2q$. Therefore, $\mathcal{C}$ has  minimum distance  $d = d'' \le 2q$.

Combining all the above discussions, the theorem holds.
\end{proof}

\bigskip
The next lemma characterizes the structural properties of  $q$-ary optimal   $(n,k,r)$-LRCs with minimum distance $d$ greater than the field size $q$.

\begin{lemma} \label{lemma:d>q}
\label{th-trivial}  
Let $k> r\ge 1$. Let $\mathcal{C}$ be a $q$-ary optimal   $(n,k,r)$-LRC with $d=n-k-\lceil{k}/{r}\rceil+2$ and $H$ be its parity-check matrix constructed in Section \ref{cpm}. If $\mathcal{C}$ has  minimum distance  $d > q$, then the dual code of $\mathcal{C}$  has minimum distance $d(\mathcal{C}^{\perp}) = r+1$, and one of the followings is true:
\begin{itemize}
  \item $r = 1$ and $2 \mid n$.
  \item $r\geq 2$, $k = r+1 = n-d$, and $\mathcal{C}$ is a near MDS code.
  \item $r\geq 2$ and $k = sr+1$, for some $s\geq 2$. In this case, 
      $(r+1) \mid n$ and the supports of the locality-rows in the parity-check matrix must be pairwise disjoint, and each locality-row has weight exactly $r+1$.
\end{itemize}
\label{thm:structure}
\end{lemma}

\begin{proof}
Since $\mathcal{C}$  has  minimum distance $d > q$,  according to  Theorem \ref{thm:distance-bound}, the parameters of $\mathcal{C}$ must satisfy $r \mid (k-1)$, which implies   $$r=1, \;\;\;\; \mbox{ or }  \;\;\;\; r \ge 2  \mbox{ and } k \bmod r = 1 .$$
For the first case of  $r=1$. Since $k$ is divisible by $r=1$, by Lemma \ref{lemma:structural-condition}, we obtain $2 \mid n$, by Lemma \ref{lemma:d-dual-r+1}, the dual code $\mathcal{C}^{\perp}$  has minimum distance    $d(\mathcal{C}^{\perp}) = r+1=2$.
Next, we discuss the case of  $r\geq 2$ and $k  \bmod r = 1$. Let $k = sr+1$ where $s\ge1$ and $H'$ be obtained from $H$ by removing any fixed $\lceil \frac{k}{r}\rceil  -1= \frac{k-1}{r} = s$ locality-rows and all the columns covered by these removed $s$ locality-rows.
By Theorem \ref{thm:subcode-mds}, $H'$ has full rank and the $[n',k',d']$ linear code $\mathcal{C}'$ with $H'$ as parity-check matrix is a $q$-ary MDS code with $d' = d$.
Since each of the removed $ s = \frac{k-1}{r} $ locality-rows has weight at most $r+1$, the code length of $\mathcal{C}'$ is $n' \ge n -  \frac{k-1}{r}   \cdot (r+1)$.
The dimension  of $\mathcal{C}'$ is
\begin{equation}\label{k'ge1}
  k' = n'-{\rm Rank}(H') \ge \left [n -  \frac{k-1}{r}   \cdot (r+1) \right ] - \left [n-k- \frac{k-1}{r} \right ] = k- \frac{k-1}{r}  \cdot r = 1.
\end{equation}
If $\mathcal{C}'$ has dimension $k'\ge 2$, then according to Lemma~\ref{lemma:defect}, the MDS code $\mathcal{C}'$ has minimum distance $d'\le q$, which implies the minimum distance of $\mathcal{C}$ is $d=d' \le q$. This contradicts the assumption that $d > q$.
Therefore, the dimension  of $\mathcal{C}'$ must be $k'=1$.
When $k'=1$, by \eqref{k'ge1}, it follows that  $n'=n -  \frac{k-1}{r}  \cdot (r+1)$.
Hence, the supports of the removed $s =  \frac{k-1}{r} $ locality-rows are pairwise disjoint and each of the removed $s$ locality-rows has weight exactly $r+1$.
Since the removed $s$ rows can be arbitrarily chosen, we  conclude  all locality-rows have weight exactly $r+1$.
Hence, there does not exist a codeword with weight less than $r+1$ in the dual code $\mathcal{C}^{\bot}$. 
Thus,  $\mathcal{C}^{\perp}$  has minimum distance  $d(\mathcal{C}^{\perp}) = r+1$.

Next, we further distinguish two cases of $s=1$ and $s\ge 2$ to discuss the structures of optimal LRCs with locality $r\ge 2$ and dimension $k=sr+1$.

If $s=1$, i.e., $k=r+1$, then $d = n-k-\lceil k/r \rceil +2 = n-k.$
The defect of $\mathcal{C}$ is $\lambda(\mathcal{C})=1$.
Since the dual code has minimum distance  $d(\mathcal{C}^{\perp}) = r+1=k$,  its defect  is $\lambda(\mathcal{C}^{\perp})=1$.
Since both the defects of $\mathcal{C}$ and  $\mathcal{C}^{\bot}$ are  $\lambda(\mathcal{C})=\lambda(\mathcal{C}^{\bot})=1$, we conclude $\mathcal{C}$ is a near MDS code.

If  $s\geq 2$,  since the removed $s$ locality-rows are arbitrarily chosen, the supports of any $s\ge2$  locality-rows with weight $r+1$ must be pairwise disjoint.
Therefore, the supports of all locality-rows in $H$ are pairwise disjoint and each has weight exactly $r+1$, and $n$ is divisible by $r+1$.

Combining all the above discussions, the lemma follows.
\end{proof}

\bigskip
For a $q$-ary optimal  $(n,k,r)$-LRC attaining the Singleton-like bound \eqref{singleton-like-bound}, since $d = n-k-\lceil k/r \rceil + 2$, the upper bound in Theorem \ref{thm:distance-bound}  gives an upper bound on  the maximal code length of a $q$-ary optimal $(n,k,r)$-LRC.

\begin{corollary} \label{length-bound}
Let $k > r \geq 1$ and $d > 2$. For  a $q$-ary optimal $(n,k,r)$-LRC $\mathcal{C}$ attaining the Singleton-like bound \eqref{singleton-like-bound}, its code length is upper bounded by
\begin{equation} \label{maximal-length-bound}
    n \leq
    \begin{cases}
      q + k + \lceil k/r \rceil -2,  & \mbox{if } r \nmid (k-1), \\
      2q + k + \lceil k/r \rceil -2,  & \mbox{if } r \mid (k-1).
    \end{cases}
\end{equation}
\end{corollary}

\smallskip

By the weight distribution of a  $q$-ary $[n,k,d]$ MDS code, we know the code length of an MDS code satisfies  $n \le q+k-1$ \cite{MacWilliamsCodingBook}, which is not tight for most cases. There is a celebrated \emph{MDS Conjecture} \cite{HuffmanCodingBook} on the maximal code length of  MDS code.
As for the maximal code length of $q$-ary optimal $(n,k,r)$-LRCs attaining the Singleton-like bound \eqref{singleton-like-bound}, the  bound \eqref{maximal-length-bound} is a very general upper bound, which might only be tight for some special cases.  It can be regarded as a counterpart of the upper bound $n\le q+k-1$ of MDS codes for the optimal LRCs attaining the Singleton-like bound.

\medskip

\section{ Enumerations of Optimal Binary LRCs  Achieving the Singleton-like Bound}\label{sec:binary-lrc}

It is well known that nontrivial binary MDS codes attaining the Singleton bound do not exist. Besides the linear codes with dimension $n$ and $0$, the only possible binary MDS codes are binary $[n,1,n]$ and $[n,n-1,2]$ codes.
In this section we will enumerate all the optimal binary $(n,k,r)$-LRCs attaining the Singleton-like bound (\ref{singleton-like-bound}) by employing the proposed parity-check matrix approach.
It is proved that in the sense of  equivalence of linear codes, there are only $5$ classes of optimal binary $(n,k,r)$-LRCs with minimum distance $d=n-k-\lceil k/r\rceil +2$. Moreover, we  enumerate all these possible $5$ classes of optimal binary LRCs by presenting their parity-check matrices.

\smallskip
In this section, $q=2$ is assumed. Suppose that $d\ge 2$, $r\ge 1$, and $k>r$ or
$\left\lceil{k}/{r}\right\rceil-1\ge 1$. Let $\mathcal{C}$ be an optimal binary $(n,k,r)$-LRC with $d=n-k-\lceil{k}/{r}\rceil+2$ and $H$ be its parity-check matrix constructed in Section \ref{cpm}. The next result follows from Theorem \ref{thm:subcode-mds}.

\begin{corollary} \label{co1-binary}
Let $H'$ be the $m'\times n'$ matrix obtained from $H$ by removing  any fixed  $\left\lceil{k}/{r}\right\rceil-1$ locality-rows and all the columns  covered by these removed locality-rows.
Then $H'$ has full rank and the binary $[n',k',d']$ linear code $\mathcal{C}'$ with  $H'$ as parity-check matrix  is a binary $[n',n'-1,2]\; (n'\ge 2)$ or $[n',1,n']\; (n'\ge 3)$ linear code.
\end{corollary}

\medskip
Next, we enumerate all  such optimal binary $(n,k,r)$-LRCs attaining the Singleton-like bound (\ref{singleton-like-bound}).

\textbf{Case 1:} $H'$ is a full-rank parity-check matrix of a binary $[n',n'-1,2] \; (n'\ge 2)$ linear code, or an all-one row vector. Moreover, by the proof of Theorem \ref{thm:subcode-mds}, $H'$ has a row with weight at most $r+1$, which implies that $H'$ has to be an all-one row with length at most $r+1$, or $n'\le r+1$.
Since $d'=2$, we have $d=d'=2$ and $n=k+\lceil k/r\rceil$.
Hence, $\mathcal{C}$ must be a binary $[k+\lceil k/r\rceil,k,2]$ linear code with locality $r$. Moreover, by (\ref{eql}), $\lceil k/r\rceil = l = n-k$, which implies that $H$ consists of only locality-rows.

If $r\mid k$, then $n=(r+1)k/r$ and $n-k=k/r$, all $k/r$ rows of $H$ must have uniform weight $r+1$ and pairwise disjoint supports. Then, in the sense of equivalence, the  $(k+k/r,k,r)$-LRC must have  parity-check matrix
\begin{equation}
\label{h31}
H= \Big(I_{\frac{k}{r}} \otimes (\underbrace{1,1,\ldots,1}_{r+1})\Big)_{\frac{k}{r}\times \frac{(r+1)k}{r}}\;,
\end{equation}
where $A \otimes B$ denotes the Kronecker product of matrices and $I_m$ denotes the $m\times m$ identity matrix.
For example, if $n=9,k=6,r=2$, the parity-check matrix of optimal binary $(9,6,2)$-LRC  is
\begin{equation*}
H=
\left(
  \begin{array}{ccccccccc}
    1 & 1 & 1 & 0 & 0 & 0 & 0 & 0 & 0 \\
    0 & 0 & 0 & 1 & 1 & 1 & 0 & 0 & 0 \\
    0 & 0 & 0 & 0 & 0 & 0 & 1 & 1 & 1 \\
  \end{array}
\right).
\end{equation*}

\smallskip
If $r\nmid k$, then $r\ge 2$. Let $k=sr+t$, where $1\le t\le r-1$, then $\lceil \frac{k}{r}\rceil = s+1$, $n=k+\lceil \frac{k}{r}\rceil=(r+1)\lceil \frac{k}{r}\rceil-(r-t)$, where $1\le r-t\le r-1$. Let $\hat H$ be a binary $\lceil \frac{k}{r}\rceil\times (r+1)\lceil \frac{k}{r}\rceil$ matrix in (\ref{h31}), where $\frac{k}{r}$ is changed to $\lceil \frac{k}{r}\rceil$.
\begin{eqnarray}
&&\mbox{$H$ is a $\lceil \frac{k}{r}\rceil\times (k+\lceil \frac{k}{r}\rceil)$ matrix obtained from $\hat H$ by deleting any $r-t$  }\nonumber\\
&&\qquad \mbox{columns of $\hat H$, such that at least one row of $H$ has weight $r+1$;}
\qquad\label{h32} \\
&&\mbox{$\underline H$ is obtained from $H$ by substituting at most $r-t$ 0's of $H$ to 1's,}\nonumber\\
&&\qquad \mbox{ such that the weight of each row of $\underline H$ is at most $r+1$.}\nonumber
\end{eqnarray}
Then, in the sense of equivalence, every $(k+\lceil k/r\rceil,k,r)$-LRC with minimum distance $d=2$ must have parity-check matrix as $H$ or $\underline H$ in (\ref{h32}). For example, if $n=10,k=7,r=3$, the parity-check matrix can be
\begin{equation*}
H=\left(
  \begin{array}{ccccccccccc}
   1 & 1 & 1 & 1 & 0 & 0 & 0 & 0 & 0 & 0 \\
   0 & 0 & 0 & 0 & 1 & 1 & 1 & 1 & 0 & 0 \\
   0 & 0 & 0 & \underline{0} & 0 & 0 & \underline{0} & \underline{0} & 1 & 1 \\
  \end{array}
\right),
\end{equation*}
where any one or two of the three zeros with underline can be substituted to $1$, and $\underline H$ is thus obtained.

\medskip

\textbf{Case 2:}  $H'$ is a full-rank parity-check matrix of a binary $[n',1,n'] \; (n'\ge 3)$ linear code.
In this case, the minimum distance of $\mathcal{C}$ is $d= d' = n' > q= 2$. $H'$ is an $(n'-1)\times n'$ matrix with $n' =d= n - k - \left \lceil {k}/{r}\right  \rceil +2$. By Theorem \ref{thm:distance-bound}, since  $\mathcal{C}$ has minimum distance  $d  > q$, it follows that $ r \mid (k-1)$. Hence,
$$r=1, \;\; \mbox{ or } \;\; r \ge 2 \mbox{ and } k  \bmod r = 1.$$

\medskip
If $r=1$, then $r \mid k$.
By Lemma \ref{lemma:structural-condition}, all locality-rows of $H$ have uniform weight $2$ and pairwise disjoint supports. Moveover, $n$ is divisible by $2$.
Let $n = 2 l$, where $l$ is the number of the locality-rows in $H$,  then $n' = d'= d = n-k-\left \lceil {k}/{r}\right \rceil+2  = 2(l-k+1).$ By Theorem \ref{thm:distance-bound}, since $r \mid (k-1)$,
\begin{equation}
 d= 2(l-k+1) \le 2q = 4,
\end{equation}
i.e., $l-k \le 1$. By $n' =2(l-k+1) \ge 3$, we have $l -k \ge 1$. Hence, $l=k+1$.
$\mathcal{C}$ must have parameters
\begin{equation}
  n=2k+2,\; k\ge2, \; r=1, \; d=4.
\end{equation}
In the sense of equivalence, the parity-check matrix of $\mathcal{C}$  has to be
\begin{equation}\label{construction-d-4-1}
    H = \left(
          \begin{array}{c}
            \quad I_{k+1} \;\otimes (1 \ 1)\\
            \hline
            \underbrace{(1 \cdots 1)}_{k+1} \otimes \;(0\  1)
             \\
          \end{array}
        \right)_{(k+2)\times (2k+2)}.
\end{equation}
For example, if $n=8,k=3,r=1$, it is
\begin{equation*}
H = \left(
  \begin{array}{cccccccc}
    1 & 1 & 0 & 0 & 0 & 0 & 0 & 0  \\
    0 & 0 & 1 & 1 & 0 & 0 & 0 & 0  \\
    0 & 0 &0 & 0 & 1 & 1 & 0 & 0  \\
    0 & 0 &0 & 0 & 0 & 0 & 1 & 1  \\
    0 & 1 &0 & 1 & 0 & 1 & 0 & 1
  \end{array}
\right).
\end{equation*}

\bigskip

If $r \ge 2$ and $k \bmod r = 1 $. Let $k = sr+1$  where $s\ge 1$. 
By Lemma \ref{lemma:d>q}, the dual code $\mathcal{C}^{\bot}$ has minimum distance $d(\mathcal{C}^{\bot})=r+1$ and all locality-rows of $H$ have uniform weight $r+1$. Next, we divide the enumeration of all such optimal binary $(n,k,r)$-LRCs into two cases: $s=1$ and $s \ge 2$.

\smallskip
For the case that $s=1$, by Lemma \ref{lemma:d>q}, we know $\mathcal{C}$ is a binary near MDS code.
By Lemma \ref{binary-near-mds}, when $k=r+1 \ge 3$ and $d=n' \ge 3$, up to the equivalence of linear codes, there exist exactly  four binary near MDS codes, whose parameters  are respectively
\begin{itemize}
  \item  the binary $[7,4,3]$  Hamming code with locality $r=3$;
  \item  the binary $[8,4,4]$  extended Hamming code with locality $r=3$;
  \item  the binary $[7,3,4]$  Simplex code with locality $r=2$;
  \item  the binary $[6,3,3]$ punctured Simplex code with $r=2$  and its parity-check matrix is
\end{itemize}
\begin{equation} \label{6-3-3-near-mds}
H = \left(
  \begin{array}{cccccc}
    0 & 1 & 1 & 1 & 0 & 0  \\
    1 & 0 & 1 & 0 & 1 & 0  \\
    1 & 1 & 0 & 0 & 0 & 1
  \end{array}
\right).
\end{equation}
It is not hard to verify that all these  four binary near MDS codes have above locality and are optimal binary LRCs attaining the  Singleton-like bound \eqref{singleton-like-bound}. In summary, their parameters are 
\begin{equation}
  n=k+d , \; 3\le k \le 4,  \; r=k-1, \;3 \le d \le 4.
\end{equation}

For the case that $s\ge 2$,  by Lemma \ref{lemma:d>q},  we know $(r+1)\mid n$ and all locality-rows of $H$ have uniform weight $r+1$ and pairwise disjoint supports. Let $l$ denote the number of locality-rows in $H$ and $n = l(r+1)$, then $n' =d' = d = n-k-\left \lceil {k}/{r}\right \rceil+2 = (l-s)(r+1)$.
By Theorem \ref{thm:distance-bound}, since   $r \mid (k-1)$,
\begin{equation}\label{binary-case2-d}
  d = (l-s)(r+1) \le 2q = 4.
\end{equation}
Since $\mathcal{C'}$ has the all-one codeword and $H'$ contains a row with weight $r+1$ which is a locality-row of $H$ before removing, we obtain that  $r+1$ must be even. Hence, $r+1 \neq 3$. Then,  $r+1 \ge 4$. By \eqref{binary-case2-d}, we have   $r+1 = 4$ and $l -s = 1$.
Since $s \ge 2$, we have $l \ge 3$. Therefore, $\mathcal{C}$ must have parameters 
\begin{equation}
  n=l(r+1)=4l, \; k=sr+1=(l-1)*3+1=3l-2,  \; r=3, \; d=4, \; l\ge 3.
\end{equation}
In the sense of equivalence, its parity-check matrix has to be
\begin{equation}\label{construction-d-4-2}
    H = \left(
          \begin{array}{c}
            I_l \otimes (1 \ 1 \ 1 \ 1)\\
            \hline
            \underbrace{(1 \ 1 \cdots 1)}_{l} \otimes \left(
                                                        \begin{array}{cccc}
                                                          0 & 0 & 1 & 1 \\
                                                          0 & 1 & 0 & 1 \\
                                                        \end{array}
                                                      \right)
             \\
          \end{array}
        \right)_{(l+2)\times 4l}.
\end{equation}
For example, if $n=12,k=7,r=3$, the parity-check matrix is
\begin{equation*}
H = \left(
  \begin{array}{cccccccccccc}
    1 & 1 & 1 & 1 & 0 & 0 & 0 & 0 & 0 & 0 & 0 & 0  \\
    0 & 0 & 0 & 0 & 1 & 1 & 1 & 1 & 0 & 0 & 0 & 0  \\
    0 & 0 & 0 & 0 & 0 & 0 & 0 & 0 & 1 & 1 & 1 & 1  \\
    0 & 0 & 1 & 1 & 0 & 0 & 1 & 1 & 0 & 0 & 1 & 1  \\
    0 & 1 & 0 & 1 & 0 & 1 & 0 & 1 & 0 & 1 & 0 & 1  \\
  \end{array}
\right).
\end{equation*}

\smallskip
Combining all of these discussions in this section, we have the following theorem.\footnote{Enumerations of optimal  binary LRCs was presented in \cite{Hao2016ISIT_BinaryLRC}, where only the first four classes of optimal codes in Theorem \ref{thm:binary-lrc} were found. In this paper, we  employ the results in \cite{Hao2018ISIT_LRCMaximalLength}, i.e., Theorem \ref{thm:distance-bound} and Lemma \ref{lemma:d>q} in Section \ref{sec:distance-bound}, to revise and simplify the analysis procedure. The proof is fixed  to enumerate  all the optimal binary LRCs, including the four specific binary  near MDS codes.}

\begin{theorem} \label{thm:binary-lrc}
Let $r\ge 1$, $k>r$ and $d\ge 2$. There are $5$ classes of optimal binary $(n,k,r)$-LRCs attaining the Singleton-like bound \eqref{singleton-like-bound}, whose parameters and parity-check matrices are respectively
\begin{itemize}
\item $(k+k/r,k,r)$, $d=2$, $k>r\ge 1$, $r\mid k$, $H$ in (\ref{h31});

\item $(k+\lceil k/r\rceil,k, r)$, $d=2$, $k>r\ge 1$, $r\nmid k$, $H$ or $\underline H$ in (\ref{h32});

\item $(2k+2,k,1)$, $d=4$, $k\ge 2$, $H$ in (\ref{construction-d-4-1});

\item $(4l,3l-2,3)$, $d=4$, $l\ge 3$, $H$ in (\ref{construction-d-4-2});

\item  $(k+d,k,k-1)$, $3 \le d \le 4, 3\le k \le 4$, $H$ of four binary near MDS codes.

\end{itemize}
In the sense of equivalence of linear codes, except these $5$ classes of optimal binary LRCs, there is no other binary $(n,k,r)$-LRC with minimum distance $d=n-k-\lceil k/r\rceil +2$.
\end{theorem}

\smallskip
\begin{remark}
For the optimal binary $(n,k,r)$-LRCs in Theorem \ref{thm:binary-lrc},  the codes in the third and fourth classes, and the binary $[7,3,4]$  Simplex code, the binary $[8,4,4]$  extended Hamming code in the fifth class  have parameters $r \mid (k-1)$ and $d=2q=4$, which attain the upper bound of minimum distance in Theorem \ref{thm:distance-bound}. The code lengths of these  optimal binary LRCs also attain the upper bound of maximal  code length in Corollary \ref{length-bound}.
\end{remark}

\medskip
\section{Conclusions}\label{sec:conclusion}

In this paper, we proposed a  systematic parity-check matrix approach  to study the bounds and constructions of $q$-ary $(n,k,r)$-LRCs.
Firstly,  simple and unified proofs for the well-known bounds of LRCs were given and
several useful structural properties  on parity-check matrices of $q$-ary optimal $(n,k,r)$-LRCs were obtained.
We derived  upper bounds on the minimum distance and maximal code length of a $q$-ary  optimal $(n,k,r)$-LRC in terms of field size $q$.
Then, by employing the parity-check matrix approach, we proved that there are only $5$ classes of possible parameters for optimal binary   $(n,k,r)$-LRCs. Moreover,  in the sense of equivalence of linear codes, we completely enumerate  all these $5$ classes of optimal binary  $(n,k,r)$-LRCs by presenting their parity-check matrices.

\bibliographystyle{IEEEtran}

%




\end{document}